\documentclass[a4paper]{article}

\usepackage{amsmath}
\usepackage{amssymb}
\usepackage{amsthm, graphicx, color}
\usepackage{algorithm,algorithmic}
\newtheorem{prop}{Proposition}

\newtheorem{thm}[prop]{Theorem}

\theoremstyle{definition}
\newtheorem{defi}[prop]{Definition}

\newcommand{\F}{\mathbb{F}}

\newcommand{\N}{\mathbb{N}}

\newcommand{\lm}{\mathrm{lm}}

\newcommand{\lt}{\mathrm{lt}}

\newcommand{\Lp}{\mathcal{L}_q(x,q^m)}

\title{ Gr\"obner Bases for Linearized Polynomials}
\author{
Margreta Kuijper and Anna-Lena Trautmann\\
Department of Electrical and Electronic Engineering\\ University of Melbourne, Australia.}
\date{}

\begin{document}

\maketitle

\begin{abstract}
In this work we develop the theory of Gr\"obner bases for modules over the ring of univariate linearized polynomials with coefficients from a finite field.
\end{abstract}

\section{Introduction}

Gr\"obner bases \cite{bu92} are a powerful conceptual and computational tool for modules over general multivariate polynomial rings. In particular, they also prove useful for modules over univariate polynomials with coefficients from a finite field. 
This motivates us to develop similar tools for modules over linearized polynomials, a special family of polynomials over a finite field, in an analogous manner. 
For more information on Gr\"obner bases for modules over finite field polynomial rings the interested reader is referred to \cite{ad94b}.

Let $\F_q$ denote the finite field with $q$ elements, where $q$ is a prime power, and let $\F_{q^m}$ denote the extension field of extension degree $m$. A \emph{linearized polynomial} over $\F_{q^m}$ is of the form
\[f(x) = \sum_{i=0}^n a_i x^{[i]}, \quad a_i \in \F_{q^m},\]
where $[i]:=q^i$. 
If the base field needs to be specified, these polynomials are also called $q$-linearized. The name "linearized" stems from the fact that linearized polynomials function as $q$-linear maps. This class of polynomials was first investigated in \cite{or33} and later on by \cite{de78}. These polynomials  have received a lot of interest in the past decades due to their application in rank-metric codes \cite{de78,ga85a} and related topics.

The set of linearized polynomials, equipped with normal polynomial addition $+$ and polynomial composition $\circ$, forms a non-commutative ring without zero-divisors (see e.g.\ \cite{ko08}). We will denote this ring of $q$-linearized polynomials over $\F_{q^m}$ by $\Lp$.

Due to the difference of composition (for linearized polynomials) and multiplication (for classical polynomials), the theory of bases in general, and Gr\"obner bases in particular, needs to be developed from scratch for the ring $\Lp$.

The paper is structured as follows: In the next section we will investigate the structure of $\Lp^\ell$ as a left module. In Section \ref{sec:Groebner} we will derive the theory of Gr\"obner bases for submodules of $\Lp^\ell$. We conclude this work in Section \ref{sec:conclusion}.

%%%%%%%%%%%%%%%%%%%%%%%%%%%%%%%%%%%%%%%%%%%%%%%%%%%%%%%%%%%%%%%%

\section{The Module $\Lp^\ell$}

As mentioned before, $\Lp$ forms a ring with addition and composition. Hence $\Lp^\ell$ forms a right or left module, which are different due to the non-commutativity of $\circ$. In this work we will consider $\Lp^{\ell}$ as a left module and investigate its left submodules. The results then easily carry over to right modules.

 Elements of $\Lp^\ell$ are of the form 
$$f:=[f_1(x) \; \dots \; f_\ell (x)] = \sum_{i=1}^\ell f_i(x) e_i $$ 
where $f_i(x)=\sum_{j} f_{ij} x^{[j]} \in \Lp$ and $e_1,\dots,e_\ell$ are the unit vectors of length $\ell$. 
To avoid confusion we denote polynomials by $f(x)$, while vectors of polynomials are denoted by $f$. If we need to index polynomials, we use the notation $f_1(x),\dots,f_s(x)$, while for vectors of polynomials we will use the notation $f^{(1)}, \dots, f^{(s)}$. 
Analogous to polynomial multiplication on  $\F_{q^m}[x]^\ell$ we define for $h(x)\in \Lp$ the left operation
\[h(x)\circ f  :=  [h(f_1(x)) \; \dots \; h(f_\ell (x))]= \sum_{i=1}^\ell h(f_i(x)) e_i  .\]
The monomials of $f$ are of the form $x^{[k]} e_i$ for all $k$ such that $f_{ik}\neq 0$.

\begin{defi}
A subset $M\subseteq \Lp^\ell$ is a \emph{(left) submodule} of $\Lp^\ell$ if it is closed under addition and composition with $\Lp$ on the left. 
\end{defi}

\begin{defi}
Consider the non-zero elements $f^{(1)}, \dots, f^{(s)} \in \Lp^\ell$. 
We say that $f^{(1)}, \dots, f^{(s)}$ are \emph{linearly independent} if for any $a_1(x),\dots, a_s(x) \in \Lp$ 
\[\sum_{i=1}^s  a_i(x) \circ f^{(i)} = [\;0 \; \dots \; 0\; ] \quad \implies \quad a_1(x)=\dots =a_s(x) = 0. \]
A generating set of a submodule $M\subseteq \Lp^\ell$ is called a \emph{basis} of $M$ if all its elements are linearly independent.
\end{defi}

One can easily see that 
\[ B= \{x e_1, xe_2 \dots, x e_\ell \}\]
is a basis of  $\Lp^\ell$, thus $\Lp^\ell$ is a \emph{free} and \emph{finitely generated} module.

We need the notion of monomial order for the subsequent results, which we will define in analogy to \cite[Definition 3.5.1]{ad94b}.

\begin{defi}
A \emph{monomial order} $<$ on $\Lp^\ell$ is a total order on $\Lp^\ell$ that fulfills the following two conditions:
\begin{itemize}
\item    $ x^{[k]} e_i  < x^{[j]}\circ (x^{[k]} e_i) $ for any monomial $x^{[k]} e_i \in \Lp^\ell$ and $j\in\mathbb{N}_{>0}$. 
\item    If $x^{[k]} e_i < x^{[k']} e_{i'}$, then $x^{[j]}\circ (x^{[k]} e_i) < x^{[j]}\circ (x^{[k']} e_{i'} )$ for any monomials $x^{[k]} e_{i}, x^{[k']} e_{i'} \in \Lp^\ell$ and $j\in\mathbb{N}_0$. 
\end{itemize}
\end{defi}
%We have different choices for monomial orders, of which the following is of interest for our investigations.

%\begin{defi}
%\begin{enumerate}
%\item
%The \emph{term-over-position (TOP) monomial order} is defined as 
%$$x^{[i_1]} e_{j_1}< x^{[i_2]} e_{j_2} :\iff i_1<i_2 \textnormal{ or }  [i_1 = i_2 \textnormal{ and } j_1<j_2 ]  .$$ 
%\item
%\end{enumerate}
%\end{defi}

%\begin{defi}
%The \emph{$(k_1,\dots,k_\ell)$-weighted term-over-position monomial order} is defined as 
%$$x^{[i_1]} e_{j_1}<_{(k_1,\dots,k_2)} x^{[i_2]} e_{j_2} :\iff $$ $$i_1+k_{j_1}<i_2 + k_{j_2} \textnormal{ or }  [i_1+k_{j_1} = i_2+k_{j_2} \textnormal{ and } j_1<j_2 ]  .$$ 
%\end{defi}

%Note that this monomial order for $\Lp^\ell$ coincides with the weighted term-over-position monomial order for $\F_{q^m}[x]$, since one could replace the $q$-degrees with normal degrees and get the classical cases.

%
%We furthermore need the following definition in analogy to the weighted term-over-position monomial order:
%\begin{defi}
%The \emph{$(k_1,\dots,k_\ell)$-weighted $q$-degree} of  $[f_1(x) \;\dots \; f_\ell (x)]$ is defined as  $\max\{k_i+ \qdeg(f_i(x)) \mid i=1,\dots,\ell\}$.
%\end{defi}
An example of a monomial order on $\Lp^\ell$ is the weighted term-over-position monomial order in \cite{ku14}.
In the following we will not fix a monomial order. The results are general and hold for any chosen monomial order.
\begin{defi}
We can order all monomials of an element $f\in\Lp^\ell$ in decreasing order with respect to some monomial order. Rename them such that $x^{[i_1]}e_{j_1}> x^{[i_2]}e_{j_2}> \dots $. Then
\begin{enumerate}
\item the \emph{leading monomial} $\mathrm{lm}(f)=x^{[i_1]} e_{j_1}$ is the greatest monomial of $f$.
\item the \emph{leading position} $\mathrm{lpos}(f)={j_1}$ is the vector coordinate of the leading monomial.
\item the \emph{leading term} $\mathrm{lt}(f)=f_{j_1, i_1}x^{[i_1]}e_{j_1}$ is the complete term of the leading monomial.
\end{enumerate}
\end{defi}

\vspace{0.2cm}

In order to define minimality for submodule bases we need the following notion of reduction, in analogy to \cite[Definition 4.1.1]{ad94b}. 

\begin{defi}
Let $ f, h \in \Lp^\ell$ and let $F=\{f^{(1)},\dots,f^{(s)}\}$ be a set of non-zero elements of $\Lp^\ell$. 
We say that \emph{ $f$ reduces to $h$ modulo $F$ in one step} if and only if
\[h= f- (( b_{1}x^{[a_{1}]}) \circ f^{(1)}+ \dots + (b_{k}x^{[a_{k}]} ) \circ f^{(k)} )\]
for some $a_{1},\dots,a_{k}\in \N_0$ and $b_{1},\dots, b_{k}\in \F_{q^m}$, where
\[\mathrm{lm}(f) = x^{[a_{i}]}\circ \mathrm{lm}(f^{(i)}), \quad i=1,\dots,k , \quad \textnormal{ and }\]
\[\mathrm{lt}(f) =( b_{1}x^{[a_{1}]})\circ \mathrm{lt}(f^{(1)})+ \dots + (b_{k}x^{a_{[k]}})\circ\mathrm{lt}(f^{(k)}).\]
We say that $f$ is \emph{minimal} with respect to $F$ if it cannot be reduced modulo $F$.
\end{defi}

\begin{defi}
A module basis $B$ is called \emph{minimal} if all its elements $b$ are minimal with respect to $B\backslash \{b\}$.
\end{defi}

\begin{prop}\label{prop:lpos}\cite{ku14p}
Let $B$ be a basis of a module $M\subseteq\Lp^\ell$. Then $B$ is a minimal basis if and only if all leading positions of the elements of $B$ are distinct.
\end{prop}
\begin{proof}
Let $B$ be minimal. If two elements of $B$ have the same leading position, the one with the greater leading monomial can be reduced modulo the other element, which contradicts the minimality. Hence, no two elements of a minimal basis can have the same leading position.

The other direction follows straight from the definition of reducibility and minimality of a basis, since if the leading positions of all elements are different, none of them can be reduced modulo the other elements.
\end{proof}

The property outlined in the following theorem is called the \emph{Predictable Leading Monomial (PLM) property}, a terminology that was introduced in  \cite{ku11} for modules in $\F_q[x]^\ell$ with respect to multiplication. For linearized polynomials it was formulated and proven in \cite{ku14j}.

\begin{thm}[PLM property,\cite{ku14j}]\label{thm:PLM}
Let $M$ be a module in $\Lp^\ell$ with minimal basis $B=\{b^{(1)},\dots,b^{(k)}\}$. Then for any $0\neq f \in M$, written as
\[f=a_1(x)\circ b^{(1)}+\dots + a_k(x)\circ b^{(k)} ,\]
where $a_1(x),\dots,a_k(x)\in \Lp$, we have
\[\mathrm{lm}(f) = \max_{1\leq i \leq k; a_i(x)\neq 0} \{\mathrm{lm}(a_i)\circ \mathrm{lm}(b^{(i)})\}  \]
where (with slight abuse of notation) $\mathrm{lm}(a_i(x))$ denotes the term of $a_i(x)$ of highest $q$-degree.
\end{thm}
%\begin{proof}
%Since $B$ is minimal, all leading positions and thus also all leading monomials of its elements are distinct (by Proposition \ref{prop:lpos}). 
%Without loss of generality assume that $\mathrm{lm}(b^{(1)})> \mathrm{lm}(b^{(2)})> \dots > \mathrm{lm}(b^{(k)})$ and that all $a_i(x)$ are non-zero. Since $\Lp$ contains no zero divisors, we have that $\mathrm{lpos}(a_i(x)\circ b^{(i)})= \mathrm{lpos}(b^{(i)})$ for $1\leq i\leq k$. As a result, all leading positions and therefore all leading monomials of the $a_i(x)\circ b^{(i)}$'s are distinct. Thus there exist $j_1,\dots,j_k$ such that
%\[\mathrm{lm}(a_{j_1}(x)\circ b^{(j_1)}) >\mathrm{lm}(a_{j_2}(x)\circ b^{(j_2)})>\dots >\mathrm{lm}(a_{j_k}(x)\circ b^{(j_k)}) .\]
%It follows that 
%\[\mathrm{lm}(f) =\mathrm{lm}(a_{j_1}(x)\circ b^{(j_1)})=\mathrm{lm}(a_{j_1}(x))\circ \mathrm{lm}(b^{(j_1)}) = \max_{1\leq i \leq k} \{\mathrm{lm}(a_i(x))\circ \mathrm{lm}(b^{(j_i)})\}  .\]
%\end{proof}

%%%%%%%%%%%%%%%%%%%%%%%%%%%%%%%%%%%%%%%%%%%%%%%%%%%%%%%%%%%%%%%%

\section{Gr\"obner Bases for Submodules of $\Lp^\ell$}\label{sec:Groebner}

We will now investigate a special family of bases, called Gr\"obner bases, for submodules of $\Lp^\ell$.

\begin{defi}\label{defi10}
Let $M\subseteq \Lp^\ell$ be a submodule.
A subset $B\subset M$ is called a \emph{Gr\"obner basis} of $M$ if the leading terms of $B$ span a left module that contains all leading terms of $M$.
\end{defi}

It is well-known that a Gr\"obner basis of a module $M$ in $\F_q[x]^\ell$ (equipped with normal multiplication) generates $M$. We will now show the analog for linearized polynomials.
\begin{thm}
Let $M$ be a module in $\Lp^\ell$ with Gr\"obner basis $B$. Then $B$ generates $M$.
\end{thm}
\begin{proof}
Let $f\in M$ and $B=\{b^{(1)},\dots,b^{(k)}\} \subset \Lp^\ell$. 
%Without loss of generality assume that $\mathrm{lm}(b^{(1)})>\dots>\mathrm{lm}(b^{(k)}) $. 
Since $B$ is a Gr\"obner basis there exist $h_1(x),\dots,h_k(x) \in \Lp$ such that 
$$\mathrm{lt}(f) = \sum_{j=1}^k h_j(\mathrm{lt}(b^{(j)})) .$$
One sees that $\mathrm{lt}(f)$ can only be a combination of the elements of the Gr\"obner basis that have the same leading position as $f$. Without loss of generality assume that this is the case for $b^{(1)},\dots,b^{(k')}, k'\leq k$. Then 
$$\mathrm{lt}(f) = \sum_{j=1}^{k'} h_j(\mathrm{lt}(b^{(j)}))= \sum_{j=1}^{k'} h_j(b^{(j)}_{m_j}x^{[m_j]}e_{\mathrm{lpos}(f)}) ,$$
where $b^{(j)}_{m_j}x^{[m_j]}e_{\mathrm{lpos}(f)}$ is the leading term of $b^{(j)}$. 
Denote $m_- := \min\{m_j \mid j=1,\dots,k'\}$. Then 
$$\mathrm{lt}(f) = \sum_{j=1}^{k'} h_j(b^{(j)}_{m_j} x^{[m_j-m_-]}(x^{[m_-]}e_{\mathrm{lpos}(f)})) $$
$$= \left(\sum_{j=1}^{k'} h_j(b^{(j)}_{m_j} x^{[m_j-m_-]} ) \right)  \circ (x^{[m_-]}e_{\mathrm{lpos}(f)})$$
and thus $x^{[m_-]}e_{\mathrm{lpos}(f)}$ symbolically divides $\mathrm{lt}(f) $. Furthermore, there exists $1\leq i\leq k'$ such that $x^{[m_-]}e_{\mathrm{lpos}(f)}=\mathrm{lm}(b^{(i)})$.

Now reduce $f$ modulo $G$ until it is minimal and call the resulting vector $r\in \Lp^\ell$. Hence there exist $h_1(x),\dots,h_k(x) \in\Lp$ such that
\[f-r = \sum_{i=1}^k h_i(b^{(i)})\]
which implies that $f-r \in M$. If $r=0$, then $f= \sum_{i=1}^k h_i(b^{(i)}) $. We will now show by contradiction that $r\neq 0$ is not possible.
If $r\neq 0$ then $r = f- \sum_{i=1}^k h_i(b^{(i)})  \in M$, since $f\in M$. Then, by the first part of the proof, there exists $h(x)\in \Lp$ and $1\leq i\leq k$ such that 
$$\mathrm{lt}(r) = h(\mathrm{lm}(g_i))$$
which means that $r$ could be further reduced modulo $G$, which contradicts the minimality assumption. Thus, we have shown that any $f\in M$ can be generated by the elements of $B$.
\end{proof}

Thus, we have shown that any Gr\"obner basis of a module is actually a basis of this module. Clearly, the other way around is not true, i.e.\ not every basis is a Gr\"obner basis, but for minimal bases the reverse implication also holds, as shown in the following.

\begin{thm}\label{thm:minimal}
Any minimal basis $B$ of a module $M\subseteq \Lp^\ell$ is a minimal Gr\"obner basis of $M$.
\end{thm}
\begin{proof}
Let $f\in M$. Since any minimal basis of a module in $\Lp^\ell $ has at most $\ell $ elements, we can assume $B=\{b^{(1)},\dots, b^{(\ell')}\}$, where $\ell' \leq \ell$. There exist $a_1(x), \dots, a_{\ell'}(x) \in \Lp$ such that $\sum_i a_i(x)\circ b^{(i)} = f$. Then by Theorem \ref{thm:PLM} 
\[\mathrm{lm}(f) = \max_{1\leq i \leq \ell'; a_i\neq 0} \{\mathrm{lm}(a_i)\circ \mathrm{lm}(b^{(i)})\}  ,\]
i.e.\ $\lm(f)$ and thus also $\lt(f)$ is in the module spanned by all $\lm(b^{(i)})$, $i=1,\dots,\ell'$.
\end{proof}

Finally, we show the existence of Gr\"obner bases of modules in $\Lp^\ell$.

\begin{thm}\label{prop:lposG}
%\begin{enumerate}
%\item
For any module $M\subseteq \Lp^\ell$ there exists a finite minimal  Gr\"obner basis. 
%Moreover, all leading positions of the elements of a minimal Gr\"obner basis are different.
%\item
%\end{enumerate}
\end{thm}
\begin{proof}
Without restriction assume that $M$ contains elements with leading position $i$ for all $i\in\{1,\dots,\ell\}$.  
Define $f_{\min, i}$ as the (non-unique) $f\in M$ with $\mathrm{lpos}(f)=i$ whose leading monomial is minimal, for $i=1,\dots, \ell$. Then $B=\{f_{\min, 1},\dots,f_{\min, \ell}\}$ forms a Gr\"obner basis of $M$, since any leading term of $M$ is an element of the module generated by the leading terms of $B$. To see this, denote an arbitrary leading term of $M$ by $c_i x^{[j_i]} e_i$ and $\lt(f_{\min, i})= c_m x^{[j_m]} e_i$; then $j_i\geq j_m$ and
\[c_i x^{[j_i]} e_i    %= \left(\frac{c_i}{c_m^{[j_i-j_m]}}x^{[j_i-j_m]}\right)\circ (c_m x^{[j_m]} e_i)
= \left(\frac{c_i}{c_m^{[j_i-j_m]}}x^{[j_i-j_m]}\right)\circ \lt(f_{\min, i}).\]
Clearly, $B$ is finite and the leading positions of all its elements are distinct.
\end{proof}

%This shows that all our decoding algorithms could also be formulated in terms of minimal Gr\"obner basis of the interpolation module.

%
%\begin{rem}
% In the context of list decoding Gabidulin codes we are only interested in modules of $\Lp^2$ with a basis consisting of two vectors, say $b^{(1)},b^{(2)}\in\Lp^2$. It can be easily seen from Definition \ref{defi10} that such a basis $\{b^{(1)},b^{(2)}\}$ is a Gr\"obner basis if and only if $\mathrm{ lpos} (b^{(1)})\neq \mathrm{ lpos}(b^{(2)})$. In fact, for this restricted special case any Gr\"obner basis is also minimal. 
% \end{rem}

%

%%%%%%%%%%%%%%%%%%%%%%%%%%%%%%%%%%%%%%%%%%%%%%%%%%%%%%%%%%%%%%%%%%

\section{Conclusion}\label{sec:conclusion}

Gr\"obner bases for modules over $\F_q[x]$ are well-known and have been extensively studied. In this work we have translated some of the definitions and results of Gr\"obner bases from the polynomial ring $\F_q[x]$, equipped with multiplication, to the linearized polynomial ring $\Lp$, equipped with composition. It turns out, that, despite the different operation used in the ring of linearized polynomials, all results covered in this work hold in both settings.

%%%%%%%%%%%%%%%%%%%%%%%%%%%%%%%%%%%%%%%%%%%%%%%%%%%%%%%%%%%%%%%%%%

\bibliographystyle{plain}
\bibliography{margreta_anna-lena}

\end{document}